%% file: PoolGame-CCS20.tex
\documentclass[sigconf,10pt]{acmart}

\usepackage{tikz-cd}

\setcopyright{none} 
\acmConference[Anonymous Submission to ACM CCS 2020]{ACM Conference on Computer and Communications Security}{Due 15 May 2019}{London, TBD}
\acmYear{2020}

\author{Dongfang Zhao}
\affiliation{%
  \institution{University of Nevada}
  \city{Reno, NV 89557}
  \country{United States}
}
\email{dzhao@unr.edu}

\begin{document}
\title{Toward Equilibria and Solvability of Blockchain Pooling Strategies: 
A Topological Approach} 

\begin{abstract}
In 2015, Eyal proposed the first game-theoretical model for analyzing the equilibrium of blockchain pooling: when the blockchain pools are abstracted as a non-cooperative game, two pools can reach a Nash equilibrium with a closed-form formula;
Moreover, an arbitrary number of pools still exhibit an equilibrium as long as the pools have an equal number of miners.
Nevertheless, whether an equilibrium exists for three or more pools of distinct sizes remains an open problem.
To this end, this paper studies the equilibrium in a blockchain of arbitrary pools.
First, we show that the equilibrium among $q$ identical pools, coinciding the result demonstrated by Eyal through game theory,
can be constructed using a topological approach.
Second, if the pools are of different size, 
we show that (i) if the blockchain's pools exhibit two distinct sizes, an equilibrium can be reached, and
(ii) if the blockchain has at least three distinct pool sizes, there does not exist an equilibrium.
\end{abstract}




\maketitle

\input{ccs-body} 

\bibliographystyle{ACM-Reference-Format}
\bibliography{ccs-sample}

\end{document}

%% file: ccs-body.tex
\section{Introduction}

In 2015, Eyal~\cite{ieyal_sp15} proposed the first game-theoretical model for analyzing the equilibrium of blockchain pooling: when the blockchain pools are abstracted as a non-cooperative game, two pools can reach a Nash equilibrium with a closed-form formula.
Moreover, the pioneering work further proved that an arbitrary number of pools still exhibit an equilibrium as long as the pools have an equal number of miners.
Whether an equilibrium exists for three or more pools of distinct sizes, however,
remains an open problem.

This paper studies the equilibrium in a blockchain of \textit{arbitrary} pools:
the number of pools might be larger than three, and each of these pools might admit any number of miners:
the pools might be of different sizes, pair-wisely.
The motivation of this study is that
relaxing the constraints on the number of pools and their size-equality would better portrait the real-world scenarios:
evidently, all production blockchain systems comprise more than two pools whose sizes are barely equal.

We approach this problem through a methodology rather than game theory though;
instead, we go for a lower-level mathematical technique---topology.
This choice of methodology is purely technical:
we found that the framework and tools in topology are more flexible and arguably, more powerful.
Although there are multiple branches of topology, such as point-set topology and algebraic topology,
their essence is very similar:
it is all about the relationship among the elements of a set and the derived topological properties thereof.
In this sense, although in the remainder of this paper we intermingle the techniques from both point-set topology and algebraic topology,
the intrinsic rationale never diverges---again, it is all about the discussion on the relationship among, this time, the miners within and across blockchain pools.

From a technical perspective, this paper will first (\S\ref{sec:equal}) show that the equilibrium among $q$ identical pools, coinciding the result demonstrated by~\cite{ieyal_sp15} through game theory,
can be constructed using a topological approach.
Instead of using game theory as a black box in a non-cooperative game, 
we show that the equilibrium must exist even if we treat the blockchain pools in a coalitional game.
The key insight of our approach, and also the approaches of other results of this paper, lies at the topological invariant on the connectivity of topological objects.
This result demonstrates the power of topology:
although an equilibrium, known as a core, of a coalitional game, might not exist,
a topological-approach can speak of more---an equilibrium does exist for this problem.

The second part of this paper's technical contribution (\S\ref{sec:sov}) is to tackle a more challenging problem in blockchain pooling, where the pools are not equal in size.
Through various point-set-topological and algebraic-topological tools,
along with a reduction to a well-known $k$-set-agreement problem in solvability,
we show that (i) if the blockchain's pools exhibit two distinct sizes, an equilibrium can be reached, and
(ii) if the blockchain has at least three distinct pool sizes,
there does not exist an equilibrium.
Once again, the results demonstrate the power of topological methods for this problem:
we can now decide whether an equilibrium exists for an arbitrary set of pools---stronger than saying ``an equilibrium might not exist."

\section{Background and Preliminaries}

\subsection{Blockchain Pools}

Due to the severe competition for block mining,
blockchain miners form coalitions, or pools, to amortize the cost.
The point is to share the reward and transaction fee with fellow miners in the same pool in a more continuous manner rather than risking a long period of a solo mining before getting anything back,
although the latter indicating a high reward with a probability.
The motivation of pooling together in blockchains is well-justified and at goodwill;
In a perfect world, every pool and miner is a good citizen and behaves as expected.

However, in the real world, there is nothing preventing one or more pools from sending a \textit{spy} miner over to other pools.
The rationale is also justifiable as follows:
if the spy can hide its local mining results,
then the target pool will be less competitive and has a lower chance of winning the race of mining blocks.
There is a downside of this attack, however:
the pool which sends out spies is also less competitive in the sense that its own computational power is cut down.
It is then not hard to see that such a dilemma might significantly affect the pooling strategies, for example,
whether to launch the infiltrating attack,
if so, what is the percentage of computation power should be devoted to the attack,
and so on.

It then becomes a natural question to ask whether such attacks among pools might reach a stable status:
pools will not launch more attacks or withdraw existing attacks---no more spies moving around,
which is called an equilibrium.
One of the pioneering works for this matter appeared in~\cite{ieyal_sp15},
which shows:
(i) no-attack is not an equilibrium: all pools are rational and would not miss the opportunity to legally impair other pools;
(ii) a blockchain of two pools does have a Nash equilibrium; and
(iii) a blockchain of $q > 2$ pools of an equal number of miners also has an equilibrium.

\subsection{Game Theory}

One key question in game theory is whether there exists an equilibrium among rational players who would not change her decision even if the decision is sub-optimal---usually implying that her utility is not maximal. 
The main reason for not switching to a higher-return decision is due to the possibility of a ``very bad'' outcome.
Quantitatively, the players calculate the expectation---a weighted average over all the possible decisions in the simplest form---and make their decisions that might not be the absolutely highest-return one.

The well-known Nash equilibrium is one such decision,
or the so-called solution concept,
in the literature of game theory.
Informally, the Nash equilibrium has a two-fold meaning.
First, the \textit{existence of Nash equilibrium} states that if an individual player has a finite number of strategies,
then there must exist an equilibrium among a finite number of such players.
Second, the \textit{Nash equilibrium of multi-player decisions}, 
as any other solution concepts of equilibrium,
once reached, would not change unless the assumptions are invalidated.
One strong assumption of Nash equilibrium is the isolation of individual players:
the players in the game will simply behave as an entity, and no collaboration is allowed.
This assumption covers many interesting problems, but obviously, not all.
The problems targeted by Nash equilibrium and its variant are called \textit{non-cooperative games}.

When players are allowed to collaborate, 
i.e., to form coalitions to work together for higher utility returns,
the game then becomes a cooperative game,
or better known as a \textit{coalitional game} in the literature of game theory.
Although coalitional games share very similar objectives and challenges as non-cooperative games, 
the solution concepts look quite different. 
For instance, the equilibrium counterpart in coalitional games is called a \textit{core},
which does not always exist.
Some approximations to the core do exist, such as the \textit{bargaining set},
but in general, a coalitional game's equilibrium point,
from existence to calculation,
remains a very challenging problem.

\subsection{Point-Set Topology}

We review the basics in point-set topology in this section.
In the literature, it is also called general topology.
Point-set topology focuses on the abstraction of relationships among elements in a set,
without constraints about the corresponding geometric objects.
Because the relationships are all defined, in an axiomatic manner,
over the abstraction of elements and sets,
the properties and theorems are applied to other branches of topology.
A detailed introduction to point-set topology can be found in textbooks such as~\cite{munkres_book00}.
The definitions we list below are by no means exhaustive;
we only sketch the ones that will be used in later chapters.
We assume the readers are familiar with basic set-theoretical terms.

\begin{itemize}
    \item \textit{Topological Space.} 
    Let $S$ be a set. A set $\mathcal{T}$ of subsets of $S$ is called a topology of $S$ as long as $\emptyset \in \mathcal{T}$ and $S \in \mathcal{T}$. The tuple $(\mathcal{T}, S)$ is called the topological space of $S$.
    If $\mathcal{T} = \{\emptyset, S\}$, then $\mathcal{T}$ is called a \textit{trivial} topology of $S$.
    If the context is clear, we also call the topology induced by $S$ \textit{topology $S$}.
    
    \item \textit{Discrete Topological Space.}
    Let $\mathcal{T}$ be a topology of set $S$.
    If for any subset $S' \subseteq S$ we have $S' \in \mathcal{T}$, then $\mathcal{T}$ is called a \textit{discrete} topology of $S$.
    Essentially, all the elements in $S$ are singleton elements in $\mathcal{T}$;
    also, any combination of those singleton elements are also in $\mathcal{T}$.
    
    \item \textit{Open Set.}
    Let $\mathcal{T}$ be a topology of $S$.
    Any element $O \in \mathcal{T}$ is called an \textit{open set}.
    The complement set $S \setminus O$ is called a closed set.

    \item \textit{Continuous Functions.}
    Let $f$ be a function from topology $\mathcal{A}$ to topology $\mathcal{B}$, 
    $f: \mathcal{A} \rightarrow \mathcal{B}$.
    We call $f$ a \textit{continuous} function if for every open set $O' \in \mathcal{B}$, there exists an open set $O \in \mathcal{A}$ such that $f(O) = O'$.
    
    \item \textit{Injective, Surjective, and Bijective Functions.}
    Let $f$ be a function from domain set $A$ to image set $B$, $f: A \rightarrow B$.
    If for any $a \in A$, $a' \in A$, and $a \neq a'$, we have $f(a) \neq f(a')$,
    then $f$ is \textit{injective}.
    If for every $b \in B$, there is an element $a \in A$ such that $f(a) = b$, then $f$ is \textit{surjective}.
    If $f$ is both injective and surjective, then $f$ is \textit{bijective}.

    \item \textit{Homeomorphism.}
    Let $f$ be a function from topology $\mathcal{A}$ to topology $\mathcal{B}$, 
    $f: \mathcal{A} \rightarrow \mathcal{B}$.
    If $f$ is both bijective and continuous, then $f$ is called a \textit{homeomorphism} of $\mathcal{A}$ and $\mathcal{B}$,
    and $\mathcal{A}$ is called \textit{homeomophic} to $\mathcal{B}$.
\end{itemize}

\subsection{Algebraic Topology}

In contrast to point-set topology, algebraic topology is more concerned with the underlying geometric realizations and algebraic structures.
One approach to study the geometric objects is to decompose them into smaller pieces, the so-called \textit{triangulation},
which does not necessarily break the object down into triangles but \textit{cells}.
If we restrict the shape of a cell to a polyhedron, i.e., only points and lines, 
then the geometric object can be depicted as a discrete approximation.
Each of these polyhedron cells is then called a \textit{simplex}, and the whole object is called a \textit{simplicial complex}.
It turns out that these concepts can also be defined in a set-theoretic context, as we will see shortly.
A delicate and beautiful theory, namely algebraic topology, has been built upon these simple yet powerful definitions, 
and this is where we start our review.
As point-set topology, we only list the concepts and notations of algebraic topology that will be used in later chapters of this paper;
a more complete and rigorous review of algebraic topology can be found in textbooks such as~\cite{munkres_book84}.

\begin{itemize}
    \item \textit{Simplicial Complex, Simplex.}
    Let $S$ be a set. Each element $v \in S$ is called a \textit{vertex}. A \textit{simplicial complex} $\mathcal{C}$ is a set of subsets of $S$ if (i) all the vertices are singleton elements in $\mathcal{C}$ and (ii) for every $X \in \mathcal{C}$,
    all subsets of $X$ are also in $\mathcal{C}$.
    Each element $\sigma \in \mathcal{C}$ is called a \textit{simplex}. 
    In the remaining of the paper, we call a \textit{simplicial complex} a complex if no ambiguity may arise.
    
    \item \textit{Face, Proper Face, Facet.}
    A \textit{face} $\tau$ of simplex $\sigma$ if $\tau \subseteq \sigma$.
    A \textit{proper face} $\tau$ of simplex $\sigma$ if $\tau \subset \sigma$.
    A simplex $\sigma$ is a \textit{facet} of a complex $\mathcal{C}$ if $\sigma$ is not a proper face of any other simplices in $\mathcal{C}$.
    
    \item \textit{Vertex Map.}
    Give two complexes $\mathcal{A}$ and $\mathcal{B}$,
    a \textit{vertex map} is a function $\varphi: V(\mathcal{A}) \rightarrow V(\mathcal{B})$,
    where $V(\cdot)$ returns the set of vertices.
    Note that we do not impose any additional properties of the function $\varphi$;
    For example, it is allowed to map two distinct vertices from $\mathcal{A}$ to the same vertex in $\mathcal{B}$.
    
    \item \textit{Simplicial Map.}
    A vertex map $\varphi$ from $V(\mathcal{A})$ to $V(\mathcal{B})$ is called a \textit{simplicial map} it carries any simplex from $\mathcal{A}$ to a simplex in $\mathcal{B}$.
    That is, if $\{a_0, \ldots, a_n\}$ is a simplex in $\mathcal{A}$, then $\{\varphi(a_0), \ldots, \varphi(a_n)\}$ is a simplex in $\mathcal{B}$.
    
    \item \textit{Carrier Map.}
    A \textit{carrier map} $\Phi$ maps each simplex in complex $\mathcal{A}$ to a subset of complex $\mathcal{B}$, usually called a \textit{subcomplex} all of which constitute a powerset denoted $2^{\mathcal{B}}$.
    To this end, the carrier map is often written as $\Phi: \mathcal{A} \rightarrow 2^{\mathcal{B}}$.
    
    \item \textit{Dimension of Simplices and Complexes $(\mathtt{dim})$.}
    The \textit{dimension of a simplex} $\sigma$ is defined as the number of its vertices minus 1: $\mathtt{dim} (\sigma) = |\sigma| - 1$.
    The reason why this is so defined is due to the geometric intuition:
    a two-vertex simplex can be geometrically viewed as a 1-dimensional line segment,
    a three-vertex simplex can be geometrically viewed as a 2-dimensional triangle, 
    and so on.
    The \textit{dimension of a complex} $\mathcal{C}$ is defined as the dimension of the highest-dimensional simplices in the complex.
    That is, $\mathtt{dim}(\mathcal{C}) = \text{sup}\{\mathtt{dim}(\sigma): \sigma \in \mathcal{C}\}$.
    
    \item \textit{Skeleton of Complexes $(\mathtt{skel}^k)$. }
    A $k$-skeleton of complex $\mathcal{C}$ is a set of simplices in $\mathcal{C}$ whose dimension is up to $k$.
    That is, $\mathtt{skel}^k \mathcal{C} = \{\sigma \in \mathcal{C}: \mathtt{dim}(\sigma) \leq k\}$.
    Following this definition, the vertices of a complex is just its 0-skeleton: $V(\mathcal{C}) = \mathtt{skel}^0 \mathcal{C}$.
    
    \item \textit{Connectivity.}
    This is one of the most important topological properties in the sense that two topological spaces cannot be equivalent (i.e., homeomorphic) if they exhibit different levels of connectivity.
    The most widely-used property, also the one used in later chapters of this paper, is called \textit{path-connected},
    which has a similar definition as graph theory.
    A complex $\mathcal{C}$ is called path-connected if any two vertices can be connected through a sequence of 1-dimensional simplices.
    If the context is clear, we simply call a complex connected if it is path-connected.
    Strictly speaking, path-connectivity is at the degree 0 of the connectivity spectrum, namely 0-connectivity.
    There is higher-degree connectivity:
    a $k$-connectivity indicates that a $(k+1)$-dimensional ball can be continuously mapped to the complex---there is no $(k+1)$-dimensional ``holes'' in the complex.
    
    \item \textit{Subdivision $(\mathtt{div})$.}
    Subdivision refers to finer decomposition of a given simplex, denoted $\mathtt{div}(\cdot)$,
    such that a new set of simplices are generated, and all of them constitute a new complex.
    Two of the most popular subdivisions are \textit{Barycentric} subdivision and \textit{Chromatic} subdivision.
    Barycentric subdivision ($\mathtt{Bary}(\cdot)$) simply takes the Barycentric center of a specific lower-dimensional skeleton, 
    joins the center to existing vertices,
    and repeats this procedure on all levels of skeletons.
    \textit{Chromatic} subdivision ($\mathtt{Ch}(\cdot)$) extends the Barycentric subdivision by introducing joint-centers such that the new vertices cannot directly share a lower-dimensional simplex with existing vertices,
    which turns to be very useful in tasks related to graph coloring.
    
    \item \textit{Joining Complexes (*).}
    In addition to dividing existing simplices, we are also able to ``glue'' together existing complexes through \textit{joining}.
    Given two complexes $\mathcal{A}$ and $\mathcal{B}$,
    the joining of two is a new complex with all the simplices from $\mathcal{A}$ and $\mathcal{B}$.
    Formally, we say $\mathcal{A} * \mathcal{B} = \{\sigma \cup \tau: \sigma \in \mathcal{A}, \tau \in \mathcal{B}\}$.
    This definition can be extended to simplices as well;
    for example, the join of a vertex and a complex is a new complex with the original vertex, the original complex, and all the edges (i.e., 1-dimensional simplices) between the original vertex and every vertex of the original complex.
    
\end{itemize}

\section{Equilibrium of Equal Pools}
\label{sec:equal}

\subsection{Problem Formulation}
\label{sec:equal_problem}

Assume there are $q$ pools, where $q \in \mathbb{Z}^+$, $q \geq 2$.
We call the \textit{size}
Two pools are said to be \textit{equal} if they have the same number of nodes, namely, \textit{size}.
Let the size of each pool be $(n+1)$, $n \in \mathbb{Z}^+$, $n \geq 1$.
Each of $q$ pools is a thus a $n$-simplex, and the input complex $\mathcal{I}$ has $q$ disconnected $n$-simplices, i.e., \textit{components}:
there is no 1- or higher-dimensional simplex among any pair of these components.
These $n$-simplices are disconnected from each other because they know nothing about others.

\begin{definition}[Vertex in $\mathcal{I}$]\label{def_vertex}
Each vertex $v$ in $\mathcal{I}$ is a tuple $(p_i^j, \bot)$,
where $p_i^j$ denotes the $i$-th pool's $j$-th node (miner) and $\bot$ denotes logical false---the miner does not join another pool other than $i$.
Evidently, we have $1 \leq i \leq q$ and $0 \leq j \leq n$.
\end{definition}

The output complex is defined as follows. 
Each vertex $v$ in $\mathcal{O}$ is a tuple $(p_i^j, k)$, where $p_i^j$ is the same as Def.~\ref{def_vertex} and $k$ denotes the $k$-th pool, $1 \leq k \leq q$.
Note that, by such definition, even if $p_i^j$ does not change its originally assigned pool, the vertex will still change the value from $\bot$ to $k$.

\begin{definition}[Auxiliary Operations on Vertices]
\label{def:vertice_operation}
By convention, we define two auxiliary operations to return the process identities and their assignments:
\begin{align*}
&\mathtt{name}(v) = p_i^j,
\\
&\mathtt{view}(v) = k.
\end{align*}
Furthermore, we define the following functions to return the miner's pool and index information:
\begin{align*}
&\mathtt{pool}(v) = i,
\\
&\mathtt{index}(v) = j.
\end{align*}
\end{definition}

Because our goal is to find an equilibrium,
each node $p_i^j$ cannot be intersected by two or more simplices in $\mathcal{O}$.
That is, if $(p_i^j, k) \in \sigma \in \mathcal{O}$ and $(p_i^j, k) \in \tau \in \mathcal{O}$,
then $\sigma = \tau$.
It follows that the output complex $\mathcal{O}$ is a disconnected complex of $N$-simplices where $N = q\cdot(n+1)-1$.
For any $N$-simplex in $\sigma \in \mathcal{O}$, we have
\[
\{\mathtt{index}(v): v\in \sigma\} = [n],
\]
where $[n]$ denotes the set $\{0, \ldots, n\}$.
Similarly, we have
\[
\{\mathtt{view}(v): v\in \sigma\} = [q].
\]

Now we are ready to define the \textit{carrier map} $\Delta$ from complex $\mathcal{I}$ to $\mathcal{O}$,
that is, $\Delta: \mathcal{I} \rightarrow 2^\mathcal{O}$.
Essentially, $\Delta$ ``deforms'' a simplex $\sigma \in \mathcal{I}$ into one possible simplex in a subset of $\mathcal{O}$, i.e., a \textit{subcomplex} $\mathcal{S} \subseteq \mathcal{O}$ in a monotonic way.
By \textit{monotinic}, we require that adding more vertices or higher-dimensional simplices into $\sigma$ would only enlarge the subcomplex $\mathcal{S}$.
If we can successfully construct a carrier map $\Delta$, 
or prove that such carrier map \textit{exists}, 
then we are guaranteed the existence of an equilibrium.

The problem of assigning pools to miners is thus represented as a triple $(\mathcal{I}, \mathcal{O}, \Delta)$.
In the remainder of this section, we will discuss how to find or construct $\Delta$, if it exists.

\subsection{Methodology Overview}

In the literature of algebraic topology, we usually face two types of problems:
(i) To demonstrate that two topological spaces are ``equivalent'' up to \textit{continuous} deformation---retaining the topological properties or the so-called \textit{topological invariant}, and 
(ii) To prove that two topological spaces cannot be equivalent.
For the first type of problem, we can either construct an explicit continuous function to map from the first space to the other or show that both spaces are \textit{homeomorphic} in the sense that their underlying \textit{algebraic} structures (e.g., groups, fields) are \textit{homomorphic}.
For the second type of problem, we are left with the only option of showing that the underlying algebraic structures are not homomorphic.

Our strategy to find the equilibrium for equal pools is as follows.
\begin{enumerate}
    \item We will construct a continuous function, i.e., a \textit{simplicial map} $\varphi$ from a specific $n$-simplex $\sigma \in \mathcal{I}$ to a $N$-simplex $\tau \in \mathcal{O}$. 
    That is, $\varphi: \mathcal{I} \rightarrow \mathcal{O}$.
    \item We will show that all the $n$-simplices in $\mathcal{I}$ are homeomorphic, such that the simplicial map is applicable to all the simplices in $\mathcal{I}$.
    \item Finally, we will prove that for any $\sigma \in \mathcal{I}$, the corresponding simplices in $\mathcal{O}$ constitute a connected subcomplex: 
    \[\displaystyle
    \bigcup_{\sigma \in \mathcal{I}} \varphi(\sigma) \subseteq \mathcal{O}.
    \]
\end{enumerate}

\subsection{Simplicial Map from $\mathcal{I}$ to $\mathcal{O}$}

First, we will rephrase the well-known results of game-theoretic blockchains in the context of algebraic topology. 
Eyal~\cite{ieyal_sp15} showed that a no-attack strategy is not a Nash equilibrium in blockchain pools.
That is to say, in the output complex $\mathcal{O}$, a connected component, which is itself a simplex, cannot have all of its nodes staying in their originally assigned pool.
We thus have our first constraint on the simplices in $\mathcal{O}$,
as stated in the following lemma.

\begin{lemma}
\label{thm:no_stay}
For any simplex $\tau \in \mathcal{O}$,
there exists a vertex $v \in \tau$ such that $\mathtt{name}(v) = p_i^j$ and $\mathtt{view}(v) \neq i$.
\end{lemma}
\begin{proof}
The result is a straightforward implication of the result in~\cite{ieyal_sp15}.
More specifically,
this can be easily verified by contradiction.
If every vertex $v \in \tau$ satisfy the following two conditions:
$\mathtt{name}(v) = p_i^j$ and $\mathtt{view}(v) = i$,
then we know that no node joins another pool in the final pool assignment.
However, this is not possible because a ``stay-at-pool'' strategy is not a Nash equilibrium\footnote{Indeed, here we assume a noncooperative game among individual miners. 
A more general and harder problem is to consider the \textit{coalition} among the miners in the same pool, essentially forming a coalitional game model. We will discuss solution concepts and solvability challenges of coalitional pooling strategies in~\S\ref{sec:sov}.} according to~\cite{ieyal_sp15},
leading to a contradiction.
\end{proof}

Recall that a simplicial map is an extended \textit{vertex map}:
not only does the map apply to $0$-simplices (i.e., vertices),
but also to any $k$-simplices, $0 \leq k \leq n$.
That is, we need to show that there exists a map from $\mathcal{I}$ to $\mathcal{O}$ such that any simplex of dimension $d$ has a corresponding $d$-dimensional simplex in $\mathcal{O}$.
Formally, we will have the following result.

\begin{lemma}
\label{thm_simplicial_map}
There exists a simplicial map from $\mathcal{I}$ to $\mathcal{O}$.
\end{lemma}
\begin{proof}
We prove this lemma by constructing a map by induction.
We will first show that a vertex map exists.
Then, we will assume such a map exists for dimension $k$,
and finally, show that a new map exists for dimension $(k+1)$ based on the $k$-dimensional map.

We define a map $\varphi_0: \mathtt{skel}^0 \mathcal{I} \rightarrow \mathtt{skel}^0 \mathcal{O}$ as follows.
Recall that $\mathtt{skel^k}$ indicates the set of simplices of dimension of up to $k$;
therefore, $\mathtt{skel}^0 \mathcal{C}$ simply represents the set of vertices of complex $\mathcal{C}$.
Recall that the dimension of $\mathcal{O}$ is much larger that of $\mathcal{I}$:
\[\mathtt{dim}(\mathcal{O}) = N = q\cdot(n+1)-1 > n = \mathtt{dim}(\mathcal{I}).
\]
Therefore, it is essential for $\varphi_0$ to map $q$ simplices from $\mathcal{I}$ to the same simplex in $\mathcal{O}$.
More specifically, we need to construct the map for all the vertices of $q$ simplices in $\mathcal{I}$.
Without loss of generality, in the following we will work on the $0$-th $n$-simplex, namely $\sigma_0$, such that $\varphi_0(\mathtt{skel}^0\sigma_0) \in \mathcal{O}$.

Recall that by definition, the output complex $\mathcal{O}$ comprises $N$-simplices with all the possible pool assignments as long as the miners are distinct.
We define the simplicial map $\varphi_0: \sigma \mapsto \tau$ based on the following two conditions:
\begin{enumerate}
    \item $\{\mathtt{name}(u): u \in \sigma \} = \{\mathtt{name}(v): v \in \tau \}$
    \item $\mathtt{view}(u) = (\mathtt{view}(v) + 1)  \mathtt{\;mod\;} q$
\end{enumerate}
Condition (1) ensures that the simplicial map $\varphi$ is \textit{name preserving},
a requirement of any protocol that can solve the problem.
Condition (2) guarantees that the resulting simplex is valid, according to Lemma~\ref{thm:no_stay}.
Intuitively, all the vertices in $\sigma_0$ move the ``next'' pool in the loop of $q$ simplices.
Note that this map is well-defined because all the pools are of equal size, as part of the assumption.
It turns out that other simplicial maps $\varphi_k$, $1 \leq k < q$ can be similarly constructed as $\varphi_0$,
and they are all well-defined due to the equal cardinality of pools.
Therefore, $\forall \sigma_k$, $0 \leq k < q$, we can always use $\varphi_0$ to map all the vertices of $q$ simplices into the same simplex in $\mathcal{O}$.
We have thus successfully built the base for the induction, i.e., $\mathtt{dim}(\mathtt{skel}^0 \mathcal{I}) = 0$.

It remains to show that higher-dimensional faces of $\sigma \in \mathcal{I}$ and $\tau \in \mathcal{O}$ can also be mapped through $\varphi$.
Suppose we have found a simplicial map $\varphi_{d-1}: \mathtt{skel}^{d-1} \sigma \rightarrow \mathtt{skel}^{d-1} \tau$, $1 \leq d \leq n$, which is indeed the case for 0-dimension as shown above,
we will need to show a $d$-dimensional varient of the simplical map $\varphi_d$ can map $\mathtt{skel}^d \sigma$ to $\mathtt{skel}^d \tau$.
Let $\sigma$ and $\tau$ be two $(d-1)$-dimensional simplices in $\mathcal{I}$ and $\mathcal{O}$, respectively.
Because $d-1 < \mathtt{dim}(\mathcal{I}) = n < \mathtt{dim}(\mathcal{O}) = N$, there must exist two vertices $u \in \mathcal{I}$ and $v \in \mathcal{O}$, respectively, such that $u \not\in \sigma$ and $v \not\in \tau$.
We construct a $d$-dimensional simplex $\sigma'$ as the joining of $u$ and $\sigma$, and do that for $v$ and $\tau$ in a similar fashion:
\begin{align*}
&\sigma' = \{u\} * \sigma,
\\
&\tau' = \{v\} * \tau.
\end{align*}
Recall that $*$ denotes the join operation between two simplices.
Evidently, we will have $\sigma' \in \mathtt{skel}^d \mathcal{I}$ and $\tau' \in \mathtt{skel}^d \mathcal{O}$,
for which the map $\varphi_d$ can be applied.

Finally, we define $\varphi$ as the composition of a series of \textit{join} operations and lower-dimensional maps $(\varphi_0, \ldots, \varphi_{n-1})$.
By construction, $\varphi$ is a simplicial map from $\mathcal{I}$ to $\mathcal{O}$.
\end{proof}

Intuitively speaking, Protocol~\ref{thm_simplicial_map} states that it is possible to merge multiple initial pool-local views into a coherent view in the final pool assignment.
Next, we will show that this property is invariant as long as the pools are kept equal in the input.

\subsection{Homeomorphic Simplices in $\mathcal{I}$}

In this section, we will show that all the $n$-simplices in $\mathcal{I}$ are \textit{homeomorphic},
meaning that all these components are essentially equivalent in the sense that one component can deform into another with continuous functions.
This property will be crucial when we prove that more than one equilibrium must exist in an equal-pool blockchain network later in~\S\ref{sec:equal_game_remark}.
The remainder of this section will demonstrate that if $\mathcal{I}$ is an equal-pool blockchain network,
then any two pools are homeomorphic.
Formally, we will prove the following lemma.

\begin{lemma}
\label{thm:homeo}
Suppose $\mathcal{I}$ is an equal-pool blockchain's input complex. For any $\sigma, \tau \in \mathcal{I}$, $\sigma$ and $\tau$ are homeomorphic. 
\end{lemma}
\begin{proof}
Here we overuse the Greek letters, e.g., $\sigma$ and $\tau$, to represent both the simplices and topological spaces.
In the latter case, the \textit{discrete} space induced by the vertices of the simplex.
That is, topological space $\sigma$ consists of all the subsets of vertices in $\sigma$:
\[
\mathcal{T}_\sigma = 2^{\mathtt{skel}^0 \sigma}.
\]
If the context is clear, we simply write $\sigma$ to indicate its topological space $\mathcal{T}_\sigma$.

We need to show that there is a function between two topological spaces,
i.e., $f: \sigma \rightarrow \tau$, such that:
\begin{enumerate}
    \item $f$ is bijective
    \item $f$ is continuous
    \item $f^{-1}$ is continuous
\end{enumerate}

Let $U \subseteq \sigma$ be any subset.
Denote its cardinality $|U| = m+1$, $-1 \leq m \leq n$, then $U$ can be rewritten as
\[
U = \{u_0, \ldots, u_m\},
\]
where $U = \emptyset$ if $m < 0$.
For each element $u \in U$, $u$ is a vertex defined in Def.~\ref{def_vertex}.
Recall that $\mathtt{index}(u)$ returns the local index of vertex $u$ in the pool identified by $\mathtt{pool}(u)$.
We extend the definition of $\mathtt{index}$ on a single vertex defined in Def.~\ref{def:vertice_operation} to a set:
\[
\mathtt{index} (U) = \{\mathtt{index} (u): u \in U\}.
\]
Evidently, since all pools are equal in size by assumption,
for any $\mathtt{index} (U)$, we can always have a $\mathtt{index} (V)$, where $V \subseteq \tau$, such that $\mathtt{index} (V) = \mathtt{index} (U)$.
This is exactly the condition based on which we define $f$:
\[
f: U \mapsto V \mathtt{\;iff\;} \mathtt{index} (U) = \mathtt{index} (V).
\]
Now, we need to verify that $f$ defined above is bijective and continuous in both directions.

To show that $f$ is bijective, we need to show that $f$ is both \textit{injective} and \textit{surjective}.
For any two subsets of $\sigma$, $U_1$ and $U_2$, $U_1 \neq U_2$.
Without loss of generality, suppose $u \in U_1$ and $u \not\in U_2$.
Then, we will have a vertex $v \in V_1 = f(U_1)$ such that $\mathtt{index}(v) = \mathtt{index}(u)$.
Since $u \not\in U_2$, we know $v \not\in V_2 = F(U_2)$,
meaning that $V_1 \neq V2$.
Therefore, we just show:
\[
U_1 \neq U_2 \implies f(U_1) \neq f(U_2),
\]
thus $f$ is injective.
Suppose $V$ is any subset of $\tau$, we will show that there must be a domain $U \in \sigma$ whose image comprises $V$---if this is true, then $f$ is surjective.
This can be trivially verified: 
because pools are all equal and the miners within each pool are following the same naming convention (cf. Def.~\ref{def:vertice_operation}), then for any set of indices $\mathtt{index}(V)$,
there must be a subset $U \subseteq \sigma$ such that $\mathtt{index}(U) = \mathtt{index}(V)$,
rendering $f$ a surjective function.
Since $f$ is both injective and surjective, we know $f$ is a bijective function.

Next, we show that $f$ is a continuous function.
That is to say, for any open set $V \subseteq \tau$,
we have an open set $U \subseteq \sigma$ such that $f(U) = V$.
Recall that an open set is defined as any subset included in the topological space.
A simplex, by definition, has all of its subsets as open sets\footnote{In topology, these open sets are also \textit{closed sets} since they are all complementary to the open sets. We do not explicitly define closed sets in the main text since we do not need closed sets in our proofs.} in the induced discrete topological space.
The continuity property can be similarly demonstrated as we proved the covering property (i.e., $f$ is surjective) before.
The indices of open set $V$ is $\mathtt{index}(V)$.
Because pools are all equal and following the same naming convention,
there must be an open set $U \subseteq \sigma$ such that $\mathtt{index}(U) = \mathtt{index}(V)$.

Finally, we need to show that the inverse function $f^{-1}: V \rightarrow U$ is also continuous.
We skip the detailed proof here because it follows exactly the same reasoning as we proved the continuity of $f: U \rightarrow V$.

Given that $f$ is both bijective and continuous, and $f^{-1}$ is also continuous, we have found a homeomorphim $f$ for $\sigma$ and $\tau$.
Since both $\sigma$ and $\tau$ are arbitrary simplices of $\mathcal{I}$, 
we conclude that all simplices in $\mathcal{I}$ are homeomorphic.
\end{proof}

\subsection{Connected Subcomplex in $\mathcal{O}$}

In the previous two subsections, we have shown that (i) there exists a simplicial map from the input complex to the output complex and (ii) all the simplices in the input complex are homeomorphic.
The last piece for our main theorem,
which will be presented in~\S\ref{sec:equal_game_remark},
is demonstrating that the mapped simplices in the output complex indeed constitute a connected subcomplex.
Formally, we need to prove the following lemma.

\begin{lemma}
\label{thm:simplex_to_complex}
Let $\varphi: \mathcal{I} \rightarrow \mathcal{O}$ be the simplicial map defined in Lemma~\ref{thm_simplicial_map}.
We define the union of the mapped simplices as $\mathcal{S}$:
\[
\mathcal{S} = \bigcup_{\sigma \in \mathcal{I}} \varphi(\sigma),
\]
then $\mathcal{S} \subseteq \mathcal{O}$ and $S$ is connected.
\end{lemma}
\begin{proof}
The first part of the proof is trivial.
We need to show that for any face $\tau \in \mathcal{S}$,
then $\tau \in \mathcal{O}$.
Because $\mathcal{S}$ is a union of $\varphi(\sigma)$,
it is suffice to show that $\varphi(\sigma) \in \mathcal{O}$.
This is evidently true because that is how $\varphi$ is constructed in Lemma~\ref{thm_simplicial_map}.

The second part, the connectedness of $\mathcal{S}$, needs more work, though.
Recall that $\varphi$ is a composition of $n+1$ maps each of which corresponds to a specific dimension $k$, $0 \leq k \leq n$.
Therefore, by construction, each $\varphi(\sigma)$ is $(n-1)$ connected:
we can always have a $n$-dimensional disk that is homeomorphic to $\varphi(\sigma)$.
It remains to show that there is at least one path connecting two simplices in $\mathcal{S}$.
We will prove this by contradiction.

Suppose there does not exist a path between two simplices $\tau_1 \in \mathcal{S}$ and $\tau_2 \in \mathcal{S}$.
Then $\tau_1$ and $\tau_2$ cannot coexist in the same simplex of $\mathcal{O}$ because all the simplices in $\mathcal{O}$ are $N$-dimensional: 
missing a path between $\tau_1$ and $\tau_2$ would have disqualified any of the two serving in the same simplex in $\mathcal{O}$.
Let $\varphi(\sigma_1) = \tau_1$ and $\varphi(\sigma_2) = \tau_2$.
It follows that $\sigma_1$ and $\sigma_2$ are not mapped into the same simplex in $\mathcal{O}$ by $\varphi$.
However, every $\varphi$ guarantees that
\[
\{\mathtt{name}(u): u \in \sigma \} = \{\mathtt{name}(v): v \in \tau \},
\]
as shown in Lemma~\ref{thm_simplicial_map}.
This means that a fixed $\varphi$ must map all the $q$ simplices in $\mathcal{I}$ into the same $N$-dimensional simplex in $\mathtt{O}$,
leading to a contradiction.
\end{proof}

\subsection{Equilibrium of Equal-Pool Blockchains}
\label{sec:equal_game_remark}

We are now ready to state the main result for equal-pool blockchains.
Specifically, we will show that there is a protocol to continuously transform the input complex into the output complex.

\begin{theorem}
\label{thm:main_equal}
If all blockchain pools are equal, then there exist protocols to allow all miners reach equilibrium.
\end{theorem}
\begin{proof}
By Lemma~\ref{thm_simplicial_map}, the simplicial map $\varphi$ maps simplices $\{\sigma_1, \ldots, \sigma_{q}\}$ in $\mathcal{I}$ into the same simplex $\tau$ in $\mathcal{O}$.
However, because $\sigma_i$'s are disconnected to each other, $1 \leq i \leq q$,
the corresponding topological space cannot continuously deform to the topological space induced by $\tau$.
That is, there does not exist a protocol induced by $\varphi$ for miners to reach equilibrium.

To fix the disconnected components, we can limit $\varphi$ to a subset of its domain---a specific simplex in $\mathcal{I}$.
Without loss of generality, let the specific simplex be the first one out of the $q$ simplices $\sigma_1$:
\[
\pi = \varphi |_{\sigma_1}.
\]
According to Lemma~\ref{thm:homeo}, all $\sigma_i$, $1 \leq i \leq q$ are homeomorphic.
Therefore, up the continuous transformation, $\varphi$ can be represented by $q$ times of $\pi$:
\[
\pi^q(\sigma_1) = \varphi(\sigma), 
\]
where 
\[\pi^q(\cdot) = \underbrace{\pi \circ \ldots \circ \pi}_{q}(\cdot).
\]

Note that $\sigma_1$ is connected by definition,
which is a $n$-dimensional simplex.
Now, we have a map $\pi: \sigma_1 \rightarrow \mathcal{O}$.
By Lemma~\ref{thm:simplex_to_complex}, we know $\pi(\sigma_1) \in \mathcal{S} \subseteq \mathcal{O}$ is connected.
Therefore, both the domain and image sets of $\pi$ are connected.
Because $\pi$ is deduced from $\varphi$ that is a simplicial map,
$\pi$ must be a simplicial map as well.
A simplicial map between two connected complexes is obviously continuous. Thus a protocol must exist to transform the input complex into the output complex.
\end{proof}

Theorem~\ref{thm:main_equal} and its proof shows that for any blockchain with equal pools,
there are at least $q$ equilibria because $\pi$ can be defined on up to $q$ distinct simplices in $\mathcal{I}$.
This is a stronger result than~\cite{ieyal_sp15} in the sense that the lower bound of the number of equilibria is elevated from one to $q$.
Nonetheless, it is arguable that equal pools might not be a practical assumption in real-world blockchain practices.
To that end, we will switch our focus on the pooling strategies where the sizes are arbitrary.

\section{Solvability of Arbitrary Pools}
\label{sec:sov}

If we relax the requirement of equal pools in blockchains,
we cannot any more reason about the equilibrium following the techniques in~\S\ref{sec:equal}:
Lemma~\ref{thm:homeo} becomes invalid.
To see this, consider one of the conditions for a homeomorphism $f$;
$f$ must be bijective between $\sigma$ and $\tau$.
However, if not all pools have the same size, then there must be two simplices $\sigma$ and $\tau$ such that:
\[
|\mathtt{skel}^0 \sigma| \neq |\mathtt{skel}^0 \tau|.
\]
Therefore, there cannot exist a bijective function between two sets of different cardinalities.
In the remainder of this section, we will show that a general blockchain-pooling problem does not have an equilibrium.
As before, we start with the problem formulation.

\subsection{Problem Formulation}

Assume there are $q$ pools, where $q \in \mathbb{Z}^+$, $q \geq 2$.
Let $\{\sigma_1, \ldots, \sigma_q\}$ denote such $q$ pools.
At least two pools have distinct numbers of miners.
Without loss of generality, let $|\sigma_1| \neq |\sigma_2|$.
In general, we pose no requirement on the equality of cardinality of other pools.
To start with an easier problem, we assume
\[
|\sigma_1| = n + 1 > m + 1 = |\sigma_k|, \; 2 \leq k \leq q.
\]
That is, we have one $n$-simplex and $(q-1)$ lower-dimensional $m$-simplices in the blockchain pools.
Moreover, all of these $q$ simplices are disconnected pair-wise.

Note that although this condition makes the problem easier than the original one, the condition is strong enough to invalidate Lemma~\ref{thm:homeo}.
Also, if we can show that this ``easier'' problem does not have an equilibrium, neither does the original problem.

As the equal-pool problem, each vertex $v$ in $\mathcal{I}$ is a tuple $(p_i^j, \bot)$,
where $p_i^j$ denotes the $i$-th pool's $j$-th node (miner) and $\bot$ denotes logical false---the miner does not join another pool other than $i$.
Evidently, we have $1 \leq i \leq q$ and $0 \leq j \leq n$.

The output complex is defined similarly. 
Each vertex $v$ in $\mathcal{O}$ is a tuple $(p_i^j, k)$, where $p_i^j$ is the same as input complex and $k$ denotes the $k$-th pool, $1 \leq k \leq q$.
Note that, by such definition, even if $p_i^j$ does not change its originally assigned pool, the vertex will still change the value from $\bot$ to $k$.

Because we are concerned with an equilibrium,
each node $p_i^j$ cannot be intersected by two or more simplices in $\mathcal{O}$.
That is, if $(p_i^j, k) \in \sigma \in \mathcal{O}$ and $(p_i^j, k) \in \tau \in \mathcal{O}$,
then $\sigma = \tau$.
It follows that the output complex $\mathcal{O}$ is a disconnected complex of $N$-simplices where $N = n + (q-1)\cdot(m+1)$.
For any $N$-simplex in $\sigma \in \mathcal{O}$, we have
\[
\{\mathtt{view}(v): v\in \sigma\} = [q].
\]
However, the $\mathtt{index}(\cdot)$ would look different for $\sigma_1$ and other simplices:
\[\displaystyle
\{\mathtt{index} (\sigma_x) \} = 
\begin{cases}
    [n],& x = 1,\\
    [m],              & \text{otherwise}.
\end{cases}
\]

We can define the carrier map $\Delta$ from complex $\mathcal{I}$ to $\mathcal{O}$, similarly to the equal-pool counterpart.
The monotonic property of $\Delta$ still holds in this arbitrary-pool problem.
We will show that such $\Delta$ cannot exist due to some topological invariant in the remainder of this section.

\subsection{Methodology Overview}

Our approach to demonstrate that an equilibrium does not exist is built upon the findings in the equal-pool counterpart presented in~\S\ref{sec:equal}.
Specifically, we will split the $q$ simplices in the input complex $\mathcal{I}$ into two subcomplexes:
(i) The first subcomplex $\mathcal{I}_1$ consists of only one simplex $\sigma_1$;
(ii) The second subcomplex $\mathcal{I}_2$ consists of $(q-1)$ $m$-simplices: $\{\sigma_2, \ldots, \sigma_q\}$.
Then, according to Theorem~\ref{thm:main_equal}, $\mathcal{I}_2$ can lead to at least $(q-1)$ equilibria, 
and we can rewrite the simplicial maps with a power of limited maps.
It remains to show that the extended problem with $\mathcal{I} = \mathcal{I}_1 \cup \mathcal{I}_2$ has no way to be mapped to a connected simplex in $\mathcal{O}$.

The main technique we use to demonstrate the impossibility is reduction.
That is, we reduce a well-known problem, the $k$-set agreement problem, whose solution does not exist for certain conditions that are satisfied by the equilibrium of an arbitrary-pool blockchain.
Specifically, we will show that we can ``simulate'' the execution of the $k$-set agreement procedure by calling the operations in an arbitrary-pool blockchain.
Indeed, if we are able to do so, that means we would be able to solve the $k$-set agreement problem,
which is not possible.
As a result, a protocol for reaching equilibrium in an arbitrary-pool network cannot exist.

\subsection{Binary Partition: Topology of $\mathcal{I}_1$ and $\mathcal{I}_2$}
\label{sec:arbitrary_binary}

We start our analysis with a simpler problem where only one simplex $\sigma_1$ has a different dimension than other $q-1$ simplices in the input complex.
Without loss of generality, we assume $\mathcal{I} = \mathcal{I}_1 \cup \mathcal{I}_2$,
and $\mathcal{I}_1 = \{ \sigma_1 \}$ and $\mathcal{I}_2 = \{ \sigma_k: 2 \leq k \leq q\}$.
Moreover, we assume $\mathtt{dim}(\mathcal{I}_1) = \mathtt{dim}(\sigma_1) = n$,
$\mathtt{dim}(\mathcal{I}_2) = \mathtt{dim}(\sigma_k) = m$,
and $n > m > 0$.

We then construct limited simplicial maps $\pi_1$ and $\pi_2$ for $\mathcal{I}_1$ and $\mathcal{I}_2$, respectively.
As before, we have the following equations:
\[\displaystyle
\begin{cases}
    \pi_1 = \varphi |_{\sigma_1},\\
    \pi_2 = \varphi |_{\sigma_2} = \ldots = \varphi |_{\sigma_q}.
\end{cases}
\]
As a result, the mapped simplices can be similarly partitioned into two parts $\mathcal{O} = \mathcal{O}_1 \cup \mathcal{O}_2$,
where
\[\displaystyle
\begin{cases}
    \mathcal{O}_1 = \{\pi_1(\sigma_1)\},\\
    \mathcal{O}_2 = \{\pi_2^{q-1}(\sigma_2)\}.
\end{cases}
\]
We can write $\mathcal{O}_2$ as a power of $\pi_2$ because all the $\sigma_k$, $2 \leq k \leq q$ are homeomorphic, as discussed before (cf. Lemma~\ref{thm:homeo}).
It turns out that we can expect two distinct sets of outcomes in the output complex.
The problem can now be represented as a triple $(\mathcal{I}_1 \cup \mathcal{I}_2, \mathcal{O}_1 \cup \mathcal{O}_2, \pi_1 \circ \pi_2^{q-1})$.
We define an auxiliary operation to return the partition of a specific simplex:
\[
\mathtt{part}(\sigma) = i
\]
such that $\mathtt{dim}(I_i) = \mathtt{dim}(\sigma)$.

\subsection{Ternary and Higher-order Partitions}
\label{sec:arbitrary_ternay}

We can naturally extend the partition techniques from two to three or more.
As we will show later in~\S\ref{sec:arbitrary_kset},
for $k \geq 3$, all the $k$-set agreement problems are not decidable to have a protocol.
In the following, we assume $3 \leq k \leq q$:
there are at least three distinct pools in terms of their sizes.
Note that we also require that the size of the largest pool is at least $k$, i.e., $n+1 \geq k$, 
because otherwise there cannot exist $k$ distinct sizes

As in the case of binary partition, we can similarly define $\mathcal{I} = \mathcal{I}_1 \cup \mathcal{I}_2 \ldots \cup \mathcal{I}_k$.
We denote the cardinality, i.e., pool size, of partitioned input complex as:
\[
|\mathcal{I}_1| = i_1,\; 
|\mathcal{I}_2| = i_2,\; 
\ldots,\;
|\mathcal{I}_k| = i_k,\; 
\]
where $1 \leq x < y \leq k \implies i_x \neq i_y$,
and require 
\[\displaystyle
\sum_{\kappa = 1}^k i_\kappa = q,\; 1 \leq \kappa \leq k.
\]
We use $K$ to denote the set of pool sizes:
\[
K = \{i_1,\; \ldots,\; i_k\}.
\]

It should be noted that $k$ partitions of the input complex do not change $\mathcal{I}$'s topological properties.
There are still $q$ components in $\mathcal{I}$,
each of which is a $i_\kappa$-dimensional simplex where $i_\kappa \in K$.
Furthermore, these $q$ simplices are disconnected since the miners are in distinct pools and learn nothing from other pools initially.

Having partitioned input complex based on their cardinalities,
we are ready to construct the output complex the corresponding carrier maps.
The partitioned simplicial maps for $k$-distinct-pool are a natural extension to the binary partition we discuss in~\S\ref{sec:arbitrary_binary}.
Recall that all simplices are homeogeneous for a specific dimension $i_\kappa \in K$.
Therefore, the compounded simplicial map can be written as
\[
O_\kappa = \pi^{i_\kappa}_\kappa (I_\kappa),\; 1 \leq \kappa \leq k.
\]
The task then can be expressed in the following triple:
\[\displaystyle
\left(
\bigcup_{1 \leq \kappa \leq k} \mathcal{I}_\kappa,
\bigcup_{1 \leq \kappa \leq k} \mathcal{O}_\kappa,
\prod_{1 \leq \kappa \leq k} \pi^{i_\kappa}_\kappa
\right),
\]
where we use multiplication to denote map composition:
\[\displaystyle
\prod_{1 \leq \kappa \leq k} \pi^{i_\kappa}_\kappa = 
\underbrace{
\overbrace{\pi_1 \circ \ldots \circ \pi_1}^{i_1} 
\circ \ldots
\overbrace{\pi_\kappa \circ \ldots \circ \pi_\kappa}^{i_\kappa} 
\circ \ldots
\overbrace{\pi_k \circ \ldots \circ \pi_k}^{i_k}
}_{q}.
\]
\subsection{Relations with $k$-set Agreement}
\label{sec:arbitrary_kset}

\subsubsection{The $k$-set Agreement Problem}

This section briefly reviews the classical $k$-set agreement (KSA) problem and its important decidability implications in the literature.
In the following sections, we will then use the KSA as a base model and reduce it to the $k$-distinct-pool (KDP) problem in blockchains.
In the literature of computation theory,
the reduction is also called ``simulation,''
in the sense that a solution protocol for the targeting problem can be wrapped up as a solution for the base problem whose decidability is already known.

\begin{definition} [$k$-set agreement task.]
\label{def:ksa}
A task of $k$-set agreement is a triple $(\mathcal{I}, \mathcal{O}, \Delta)$,
where the input $\mathcal{I}$ is comprised of multiple $n$-simplices,
the output $\mathcal{O}$ is a complex of $k$ pairwise-disconnected simplices each of which is connected by itself,
and the carrier map $\Delta$ mapping each $n$-simplex in $\mathcal{I}$ to a subcomplex in $\mathcal{O}$. 
The initial views of vertices in $\mathcal{I}$ are taken from a set $V^{in}$,
the final views of vertices in $\mathcal{O}$ are taken from a set $V^{out}$,
whose cardinality is $k$, i.e., $|V^{out}| = k$.
While the values in $\{\mathtt{view}(\sigma): \sigma \in \mathcal{I}\}$ can be arbitrary,
the output complex is required to have exactly $k$ simplices, each of which takes a distinct value from $V^{out}$.
\end{definition}

Intuitively, KSA requires that the outcome can have up to $k$ distinct values for a group of $(n+1)$ members.
It is a generalization of the well-known \textit{consensus} problem where only one value can be chosen by the members.
Indeed, the consensus problem is a degenerate of KSA where $k = 1$.
Although KSA may seem easier than the consensus problem,
KSA is far from trivial under the regular computation models.
A computation model is comprised of multiple facets, such as \textit{failure models} (e.g., crash failures, Byzantine failures),
\textit{communication models} (e.g., shared memory, message passing),
and \textit{timing models} (e.g., synchronous processes, asynchronous processes).
The literature~\cite{mherlihy_book13} showed that a model is capable of solving a $k$-set agreement task only when $k \leq 2$.
To make this paper self-contained, we restate it here as the following lemma.

\begin{lemma} [KSA is solvable if $k \leq 2$]
\label{thm:ksa}
A model is capable of solving a $k$-set agreement task only if $k \leq 2$.
\end{lemma}
\begin{proof}
See Theorem 5.6.14 in~\cite{mherlihy_book13}. 
\end{proof}

In the remainder of this section,
we call a KSA problem with $k=2$ and $k=3$ 2SA and 3SA, respectively.
Similarly, we call a KDP problem with $k=2$ and $k=3$ 2DP and 3DP, respectively.
In~\S\ref{sec:arbitrary_ksa_2sa}, we will show that a 2SA protocol can reduce to a 2DP protocol,
which means a blockchain with two distinct pool sizes can reach an equilibrium.
In~\S\ref{sec:arbitrary_ksa_3sa}, we will show that a 3DP protocol can reduce to a 3SA protocol,
meaning that a blockchain with three or more distinct pool sizes does not have an equilibrium.

\subsubsection{2SA reduces to 2DP}
\label{sec:arbitrary_ksa_2sa}

We first provide the basic concepts and notations of computing theory in the context of reduction,
much of which was elaborated in~\cite{mherlihy_book13}.

A \textit{task} is a problem to be solved. 
Examples of tasks include finding the equilibrium of equal pools discussed in~\S\ref{sec:equal_problem}.
For completeness, we formally define the task here.

\begin{definition} [Task]
A \textit{task} is a triple $(\mathcal{I}, \mathcal{O}, \Delta$),
where $\mathcal{I}$ is the input complex comprising all the possible input simplices,
$\mathcal{O}$ is the output complex comprising all the valid results represented by subcomplexes,
and $\Delta$ is a carrier map between the input and output complexes and satisfies monotonic property (i.e., closed on inclusion relationship):
\[
\Delta: \mathcal{I} \rightarrow 2^\mathcal{O}.
\]
\end{definition}

A \textit{protocol} is an algorithm that solves the task, i.e., a procedure to reach a subset of the valid subcomplex in $\mathcal{O}$.
A protocol also takes the same input complex $\mathcal{I}$, and follow the application-specific procedure to end up in a \textit{protocol complex}.
The course of ``following the application-specific procedure'' is abstracted also by a carrier map, namely $\Xi$.
Formally, a protocol is defined as follows.

\begin{definition} [Protocol]
A protocol is a triple $(\mathcal{I}, \mathcal{P}, \Xi)$,
where $\mathcal{I}$ is the input complex,
$\mathcal{P}$ is a complex representing all the possible final results according to the algorithm,
and $\Xi$ is a carrier map:
\[
\Xi: \mathcal{I} \rightarrow 2^\mathcal{P}.
\]
\end{definition}

We say that a protocol solves a task if the results in $\mathcal{P}$ are encompassed by the expected outcomes $\mathcal{O}$,
as stated in the following.

\begin{definition}
A protocol $(\mathcal{I}, \mathcal{P}, \Xi)$ solves a task $(\mathcal{I}, \mathcal{O}, \Delta$ if there is a simplicial map $\delta$ 
\[
\delta: \mathcal{P} \rightarrow \mathcal{O}
\]
such that for any simplex $\sigma \in \mathcal{I}$ the following is satisfied\footnote{This property is called ``carried by'' in the literature. We here avoid using this terminology for less confusion between \textit{carried by} and \textit{carrier map}.}:
\[
\delta (\Xi (\sigma)) \in \Delta (\sigma).
\]
By convention, we also write $\delta \circ \Xi \subseteq \Delta$,
and $\delta$ is usually called a \textit{decision map}.
\end{definition}
The relationship between protocols and tasks can be illustrated in the following diagram in the sense of category theory.

\begin{equation*}
\begin{tikzcd}
\mathcal{I}    \arrow [r, "\Xi"]
            \arrow [dr, swap, "\delta \circ \Xi\; \subseteq\; \Delta"]
&
\mathcal{P}   
    \arrow [d, "\delta"]
\\
{}&
\mathcal{O}   
\end{tikzcd}
\end{equation*}

Having formalized the tasks and protocols, 
we are ready to define a computation model,
as follows.
Intuitively, a model is just a collection of protocols that are applicable to a specific input complex.

\begin{definition} [Computation Model]
A computation model on an complex $\mathcal{I}$ is a collection of triples $(\mathcal{I}, \mathcal{P}_i, \Xi_i)$, $i \geq 0$.
\end{definition}

Now we are ready to portrait the relationship between computation models.

\begin{definition} [Model Reduction]
Suppose there are two computation models, $\mathbb{R}$ and $\mathbb{V}$.
By convention, they are called \textit{real model} and \textit{virtual model}, respectively.
Real model $\mathbb{R}$ is said to \textit{reduce} to virtual model $\mathbb{V}$ if every protocol for $\mathbb{V}$ implies a protocol in $\mathbb{R}$.
\end{definition}

Usually, the reduction is implemented by a \textit{simulation} of one protocol in another,
as defined below.

\begin{definition} [Model Simulation]
\label{def:model_sim}
Let $(\mathcal{I}, \mathcal{P}_r, \Xi_r)$ be a protocol for model $\mathbb{R}$ and $(\mathcal{I}, \mathcal{P}_v, \Xi_v)$ be a protocol for model $\mathbb{V}$, respectively.
We call $\mathbb{V}$ simulates $\mathbb{R}$ if there exists a simplicial map $\Phi$
\[
\Phi: \mathcal{P}_r \rightarrow \mathcal{P}_v.
\]
The map $\Phi$ is called a \textit{simulation} between $\mathbb{R}$ and $\mathbb{V}$.
For any $\sigma \in \mathcal{I}$, we have 
\[
(\Phi \circ \Xi_r)(\sigma) \subseteq \Xi_v(\sigma).
\]
\end{definition}
From the perspective of category theory, the morphism can be illustrated in the following diagram.

\begin{equation*}
\begin{tikzcd}
\mathcal{I}    \arrow [r, "\Xi_r"]
            \arrow [dr, swap, "\Phi \circ \Xi_r\; \subseteq\; \Xi_v"]
&
\mathcal{P}_r   
    \arrow [d, "\Phi"]
\\
{}&
\mathcal{P}_v   
\end{tikzcd}
\end{equation*}

We construct the simulation from 2SA to 2DP as follows.
Suppose the 2SA is represented by task $(\mathcal{I}, \mathcal{O}, \Delta)$ and has a protocol $(\mathcal{I}, \mathcal{P}, \Xi)$.
We know such a protocol must exist because when $k = 2$,
the $k$-set-agreement problem turns to be a graph-theoretical problem and there are many algorithms (e.g., breadth-first, depth-first) to decide whether a pair of vertices ($u$, $v$) has a path between $\Delta(u)$ and $\Delta(v)$.

For the 2SA task, assume there are $N$ nodes in total,
each of which is a vertex in $\mathcal{I}$.
For its protocol $(\mathcal{I}, \mathcal{P}_r, \Xi_r)$,
every time the protocol complex $\mathcal{P}_r$ grows into the next step,
the new complex is a subdivision of the previous complex.
There must exist a sufficiently large integer $M$ such that
\[
\mathtt{div}^M \mathcal{I} = \mathcal{P}_r,
\]
where $\Xi = \mathtt{div}^M$.
Therefore, the protocol for 2SA can be rewritten as a triple
$(\mathcal{I}, \mathtt{div}^M \mathcal{I}, \mathtt{div}^M)$.

For each $\mathtt{div}$ operation, we map it to the swap of pool assignment for a set of miners.
We denote $\mathtt{swap}^M$ the following operation:
\[
\mathtt{pool}(\sigma) = (\mathtt{pool}(\sigma) + 1) \;\mathtt{mod}\; q, 
\]
such that 
\[
\mathtt{index}(\sigma) = M \;\mathtt{mod}\; |I_{\mathtt{part}(\sigma)}|.
\]
Intuitively, we restrict the miners to switch to a new pool in a systematic manner.
It should be clear that this manner is only a subset of all the possible movements;
the point is that these movements are all allowed by $\mathcal{P}_v$, and any final assignment can be achieved by this operation as long as $M$ is sufficiently large.
As a result, the protocol for 2DP $(\mathtt{I}, \mathcal{P}_v, \Xi_v)$ can be stated as 
$(\mathcal{I}, \mathtt{swap}^M \mathcal{I}, \mathtt{swap}^M)$.

It remains to show that for any simplex $\sigma \in \mathtt{div}^M \mathcal{I}$, there exists a $\Phi$ such that $\Phi(\sigma) \in \mathtt{swap}^M \mathcal{I}$.
Note that for any $\sigma \in \mathtt{div}^M \mathcal{I}$,
the vertices in $\sigma$ are still $\{v_0, v_1, \ldots, v_n\}$,
where $n+1$ is the total number of miners.
Of course, the simplex $\sigma$ also contains all the higher-dimensional faces on the vertices.
Therefore, we need to show that for sufficiently large $M$,
the resulting simplex consists of all $n+1$ vertices.
As we will show in the following lemma, 
this is indeed the case.

\begin{lemma}
\label{thm:arbitrary_fullsimplex}
For sufficiently large $M \in \mathbb{Z}^+$, There is at least one simplex $\tau \in \mathtt{swap}^M \mathcal{I}$ such that $\mathtt{dim}(\tau) = \mathtt{dim}(\mathtt{I})$. 
\end{lemma}
\begin{proof}
Suppose for contradiction that there is no simplex of dimension $\mathtt{dim}(\mathcal{I})$.
Recall that in 2DP, there are only two partitions of pools with distinct number of miners.
It follows that there is at least one ``missing'' vertex $v$ in a specific pool that is not exchanged by any pool of the other partition after $M$ times of pool swaps.
However, the swap operation is defined in such a round-robin fashion that,
as long as $M$ is larger than the pool size,
all the vertices would be swapped.
Since we assume $M$ is ``sufficiently large,''
it is larger than the pool size also,
thus leading to a contradiction.
\end{proof}

We are now ready to make the first main conclusion for arbitrary-pool-size blockchains, as stated in the following theorem.

\begin{theorem}
\label{thm:arbitrary_2pool}
If a blockchain has pools of only two distinct sizes,
then it has an equilibrium.
\end{theorem}
\begin{proof}
If the blockchain does not has an equilibrium on the pools,
then according to Lemma~\ref{thm:arbitrary_fullsimplex} there cannot be a protocol for the 2-set-agreement problem.
However, it is known that there exist algorithms solving the 2-set-agreement problem.
\end{proof}

\subsubsection{KDP reduces to KSA, $k \geq 3$}
\label{sec:arbitrary_ksa_3sa}

In contrast to the previous section,
we will show that when the distinct pools are larger than or equal to three,
there does not exist an equilibrium.
In the remainder of this section, we will assume $k \geq 3$,
and both KDP and KSA assume the number of underlying distinct sets/pools is at least three.
Although the reduction approach is implemented through simulation,
we show that this simulation can be constructed in the inverse direction:
any protocol solving a KSA implies a protocol in KDP.
If this is the case, and since we know that KSA's protocol is undecidable (again, assuming $k \geq 3$),
KDP cannot have an equilibrium when $k \geq 3$.
As before, we start with constructing the protocols for KSA and KDP, respectively.

Suppose a protocol $(\mathcal{I}, \mathcal{P}_v, \Xi_v)$ solves the KDP problem.
Based on the discussion in~\S\ref{sec:arbitrary_ternay},
the protocol complex is a union of $k$ subcomplexes:
\[
\mathcal{P}_v = \bigcup_{1 \leq \kappa \leq k} \mathcal{P}^\kappa_v,
\]
where
\[
\mathtt{dim}(\mathcal{P}_v^\kappa) = |I_\kappa|,\; 1\leq \kappa \leq k,
\]
and $I_\kappa$ is the $\kappa$-th partition---the set of $|I_\kappa|$-simplices in the input complex $\mathcal{I}$.
While we do not know the operational progress of $\mathcal{P}_v$ because it depends on the specific algorithm, 
it is evident that the final complex must have such a $n$-simplex that is disconnected from any other simplices because otherwise some vertices would have ambiguous views of values---not an equilibrium.
Since the protocol complex is split into $k$ partitions,
there must be $k$ distinct $n$-simplices in $\mathcal{P}_v$.

Now, we switch our focus to the protocol $(\mathcal{I}, \mathcal{P}_r, \Xi_r)$ that solves the KSA problem.
It turns out that although the output complex is different for $k = 2$ and $k = 3$, 
the protocol complex remains unchanged:
the input complex is still subdivided\footnote{In the literature, \textit{Barycentric} subdivision, denoted by $\mathtt{Bary(\cdot)}$, and \textit{Chronmatic} subdivision, denoted by $\mathtt{Ch(\cdot)}$, are the most widely-used subdivisions. Here, we ignored the details and simply say $\mathtt{div}(\cdot)$.} repeatedly.
That is, the protocol can still be written $(\mathcal{I}, \mathtt{div}^M \mathcal{I}, \mathtt{div}^M)$ for sufficiently large $M \in \mathbb{Z}^+$.

We are ready to construct the simplicial map $\Phi$ (cf. Def.~\ref{def:model_sim})
\[
\Phi: \mathcal{P}_v \rightarrow \mathtt{div}^M \mathcal{I}.
\]
We need to first map the $k$ distinct $n$-simplices to simplices in $\mathtt{div}^m \mathcal{I}$ for some $m \geq 1$ that are disjoint.
It turns out that $m = 1$ can suffice:
$\mathtt{div} \mathcal{I}$ is capable to hold $k$ such $n$-simplices.
To demonstrate this, we need the following lemma.

\begin{lemma}
\label{thm:arbitrary_subdivision}
An $n$-simplex can be subdivided into $(n+1)$ disjoint $n$-dimensional simplices.
\end{lemma}
\begin{proof}
We choose \textit{Chromatic} subdivision, $\mathtt{Ch(\cdot)}$, to prove the claim.
In the following we provide an informal definition of Chromatic subdivision;
a full explanation can be found in~\cite{kozlov_book08}.
A Chromatic subdivision recursively divides the $n$-dimensional simplex's $k$-skeletons, $0 \leq k < n$, into $k$ pieces plus the $\epsilon$-areas.
For a 1-dimensional line segment,
the $\epsilon$-area is a central segment confined by two new vertices in-between the original two endpoints.
For a 2-dimensional triangle,
an inner triangle with its three vertices is ``locally'' connected to the two new vertices in the 1-dimensional case and the three vertices in the original triangle.

Let $\sigma$ be a $n$-dimensional simplex.
For each facet of $\mathtt{skel}^{n-1} \sigma$, 
there must be $(n-1)$ new distinct vertices,
which together with the vertex in the central $n$-dimensional inner simplex constitute a local $n$-simplex.
Since there are exactly $(n+1)$ of $(n-1)$-skeletons,
we can find $(n+1)$ of these local $n$-dimensional simplices,
none of which is the original $n$-simplex.
Denote these simplices as 
$
\{\sigma_0, \sigma_1, \ldots, \sigma_n\}.
$
Evidently, $\mathtt{dim}(\sigma_k) = n$, $0 \leq k \leq n$.
For any $\sigma_k$, none of its vertices is shared with others because $n$ vertices are from its local $(n-1)$-skeleton and the $(n+1)$-th vertex is from the central inner $n$-simplex that is distinct for each local simplex.
Similarly, none of $\sigma_k$'s higher skeletons are shared with any other $\sigma_{k'}$, $k' \neq k$, because these skeletons are encompassed by a ``larger'' simplex contributed by the original vertices in $\sigma$.
Therefore, we have 
\[
\bigcup_{0 \leq k < k' \leq n} \sigma_k \cap \sigma_{k'} = \emptyset,
\]
which is claimed by the theorem.
\end{proof}

Now we are ready to state the main result of the solvability of $k$-distinct-pool problem, $k \geq 3$.

\begin{theorem}
\label{thm:arbitrary_3pool}
If a blockchain has pools of at least three distinct sizes,
it does not have an equilibrium.
\end{theorem}
\begin{proof}
First, we show that there is a simulation from KDP to KSA.
According to Lemma~\ref{thm:arbitrary_subdivision},
$\mathtt{div}^1 \mathcal{I}$ generates $(n+1)$ disjoint $n$-dimensional simplices\footnote{As other algebraic operations, $\mathtt{div}^1(\cdot)$ is also simplified as $\mathtt{div}(\cdot)$ in the literature. Here we explicitly spell it out to emphasize its power.}.
It should be noted that $k < n$, 
because otherwise we cannot have more partitions than the total number of miners.
It follows that $\mathtt{div}^1 \mathcal{I}$ has enough capacity to accommodate $k$ distinct pools,
each of which cannot have more than $(n+1)$ vertices.
It remains to show that all the simplices in $\mathcal{P}_v$ other than the $k$ simplices of distinct pools can also be mapped to $\mathtt{div}^M \mathcal{I}$.
This is a trivial task because we can always keep subdividing the simplices that are not chosen to hold the $k$ disjoint simplices in $\mathtt{div}^1 \mathcal{I}$ until we have enough capacity.

We know that KSA cannot be solved for $k \geq 3$.
Because KDP can simulate KSA, KDP cannot be solved either.
\end{proof}

\subsection{Blockchains with Arbitrary Pools}

In this section, we summarize the main results for blockchains with arbitrary pools and discuss the relations with game-theoretical approaches.

If a blockchain has pools of only two distinct sizes,
Theorem~\ref{thm:arbitrary_2pool} states that an equilibrium indeed exists.
It should be clear, however, that this claim says nothing about how the protocol's complexity looks like;
the algorithm could be polynomial, exponential, or NP-complete, which is not the scope of this paper.
In practice, a two-distinct-pool blockchain network is, admittedly, rare.
To this end, it is more interesting to draw some conclusions on the case where more than two pools are formed in the blockchain.
Unfortunately, Theorem~\ref{thm:arbitrary_3pool} tells us that no equilibrium exists for these more interesting cases.
Of course, the impossibility result for $k \geq 3$ distinct pools, although somewhat discouraging,
is only based on the mainstream computing models nowadays.
The authors, in their very humble opinions, still hold high hopes that future hardware and communication advances plus possibly more advanced mathematical tools would lead to equilibrium results.

Speaking of equilibria over groups of participants, we just cannot skip the discussion on game theory.
From a game-theoretical point of view, the pooling strategies are aligned with the so-called \textit{coalitional games}.
Unlike \textit{non-cooperative games}, where one is interested in the utility equilibrium for individual participants,
a coalitional game deals with the utility of a coalition of participants---just like the pools of miners in blockchains.
More importantly, the well-known \textit{Nash equilibrium} that states an equilibrium \textit{always} exists is only applicable to non-cooperative games;
the counterpart equilibrium in coalitional games, 
namely \textit{core}, is proven not to always exist.
In this sense, a large portion of this section just tries to achieve a stronger result than that in game theory:
when $k \geq 3$, such a core does not exist for $k$-distinct-pool blockchains.
In coalitional game theory, there are a few more relaxed equilibria than the core per se.
For instance, a \textit{bargaining set} might be a practical replacement that implies that a participant not in the core might still choose to stay in its current coalition because its claims (for higher utility, for example) can be countered.
An extensive review of a coalitional game's solution concepts can be found in~\cite{peleg_book07}.
For a general introduction to game theory, see~\cite{maschler_book13}.

\section{Additional Related Work}

We briefly review important work on game-theoretical techniques and topological approaches in blockchains, distributed computing, and computer security.

After the pioneering work~\cite{ieyal_sp15} on the game-theoretical analysis of blockchain pools,
game-theoretical approaches started to revive in computer security and blockchain communities.
In~\cite{mnasr_ccs16}, game theory is used to model the interactions between decoy router deployers and the censors.
In~\cite{ieyal_ccs18}, authors studied miner's behavior and blockchain's status if the miners' incentive lies only on transaction fees.
In~\cite{dzhao_blockchaingame20}, a game-theoretical model was developed by considering Byzantine failures for permissioned blockchains.

Topological approaches have a long story in distributed computing and are recently used to design cross-blockchain transactions~\cite{dzhao_cidr20},
such as distributed cross-blockchain protocols~\cite{dzhao_vldb20}.
In~\cite{dzhao_focs20}, a topological space is constructed to be homeogeneous to the cross-blockchain transactions using point-set topology.
In~\cite{dzhao_disc20}, an algebraic-topological model shows that a naive message-passing protocol cannot complete a cross-blockchain transaction.
In~\cite{tnowak_podc19}, point-set topology was used to characterize the solvability of consensus problems in distributed systems.
Recent work on modeling the data structure of blockchains through algebraic topology appears at~\cite{wfloch_blockchainhomology18}. 

\section{Conclusion}
Our conventional wisdom on blockchain's pooling equilibrium is limited:
when the blockchain pools are abstracted as a non-cooperative game, two pools can reach a Nash equilibrium with a closed-form formula;
an arbitrary number of pools still exhibit an equilibrium as long as the pools have an equal number of miners.

This paper studies the equilibrium in a blockchain of arbitrary pools.
First, we show that the equilibrium among $q$ identical pools, coinciding the result demonstrated by~\cite{ieyal_sp15} through game theory,
can be constructed using a topological approach.
Second, if the pools are of different size, 
we show that (i) if the blockchain's pools exhibit two distinct sizes, an equilibrium can be reached, and
(ii) if the blockchain has at least three distinct pool sizes, there does not exist an equilibrium.

\section*{Acknowledgement}

This work is, in part, supported by the U.S. Department of Energy under contract number {DE-SC0020455}.
This work is also supported by a Google Cloud award and an Amazon research award.
The authors are grateful for the valuable discussion with Mohammad Sadoghi (University of California, Davis) on an earlier version of this work.